\documentclass[11pt]{article}
\usepackage[T1]{fontenc}
\usepackage{graphicx}
\usepackage{a4wide}
\usepackage{hyperref}
\usepackage{color}

\usepackage{todonotes}
\usepackage{amsmath}
\usepackage{amssymb}
\usepackage{amsthm}
\usepackage[ruled,vlined,linesnumbered]{algorithm2e}
\SetKw{KwNot}{not}
\SetKw{KwOR}{or}
\SetKw{KwAND}{and}
\SetKw{KwExitFor}{exit for-loop}
\SetKw{KwStep}{Step}

\newcommand{\Rset}{\mathbb{R}}

\newcommand{\Uset}{\mathbb{U}}

\usepackage{authblk}

\newtheorem{theorem}{Theorem}

\begin{document}
%
%
%
%
%

\title{Solving the recoverable robust shortest path problem in DAGs}

\author[1]{Marcel Jackiewicz}
\author[1]{Adam Kasperski\footnote{Corresponding author}}
\author[1]{Pawe{\l} Zieli\'nski}

\affil[1]{
Wroc{\l}aw  University of Science and Technology, Wroc{\l}aw, Poland\\
            \texttt{\{marcel.jackiewicz,adam.kasperski,pawel.zielinski\}@pwr.edu.pl}}

\maketitle              
\begin{abstract}
This paper deals with the recoverable robust shortest path problem
under the interval uncertainty
representation. The problem is known to be   strongly 
NP-hard and not approximable in  general digraphs.
Polynomial time algorithms for 
the problem under consideration in DAGs are proposed.
\end{abstract}

\section{Introduction}
In the \emph{shortest path problem}  (\textsc{SP}),
we are given a \emph{multidigraph} $G=(V,A)$ that
 consists
of a finite set of nodes~$V$, $|V|=n$,  and a finite multiset of
arcs~$A$, $|A|=m$. A deterministic cost is associated  with each arc, two nodes  $s\in V $ and $t\in V$ are
distinguished as the \emph{starting node} and the
\emph{destination} node, respectively. In particular, node~$s$ is
called \textit{source} if no arc enters $s$ and node~$t$ is called
\textit{sink} if no arc leaves~$t$. Let $\Phi$ be the set of all simple $s$-$t$ paths in~$G$.  
We wish to find
  a simple $s$-$t$ path in $G$        
with
the minimum total cost. The shortest path problem can be solved efficiently using several polynomial time algorithms (see, e.g.,~\cite{AMO93}). 

In the \emph{recoverable robust} version of~\textsc{SP} (\textsc{Rob Rec SP} for short), 
 we are given: the \emph{first-stage arc costs}  $C_e\geq 0$, $e\in A$, which are precisely known;
the set $\mathbb{U}$ of possible realizations
of uncertain second-stage arc costs, called \emph{scenarios}, represented
in this paper by the Cartesian product of interval costs, namely,
$\Uset=\{S=(c_e^S)_{e\in A}\,:\, c_e^S\in [\hat{c}_e, \hat{c}_e+\Delta_e], e\in A\} \subset \Rset_{+}^{|A|}$,
where $\hat{c}_e$ is a \emph{nominal arc cost} of $e\in A$ and $\Delta_e\geq 0$ is the maximum 
deviation of the cost of~$e$ from its nominal value;
and a
\emph{recovery parameter}~$k$, $0\leq k < |V|$.
A decision process in \textsc{Rob Rec SP} is 2-stage and consists in finding a path~$X\in \Phi$
 in the first stage. Then in the second (recovery) stage, after a cost scenario~$S=(c_e^S)_{e\in A}$ reveals,
the first stage path~$X$ can be modified by choosing a cheapest path~$Y$, under the revealed
cost scenario~$Y$, from its \emph{neighborhood}~$\Phi(X,k)$ depending on the recovery parameter~$k$
(a limited recovery action is allowed). Thus,
the goal is to find a pair of paths $X^*\in \Phi$ (the first stage path) and $Y^*\in \Phi(X^*,k)$
  (the second stage path), which minimize the sum of total costs of~$X^*$ and~$Y^*$
 at the worst case, i.e. 
 \begin{equation}
\textsc{Rob Rec SP}: \; \min_{X\in \Phi} \left(\sum_{e\in X} C_e + \max_{S\in \Uset}\min_{Y\in \Phi(X,k)} \sum_{e\in Y} c_e^S\right).
\label{rrsp}
\end{equation}
We examine here
the neighborhood 
$\Phi(X,k)=\{Y\in\Phi \,:\,  |Y\setminus X|\leq k\}$,
called  the \emph{arc inclusion  neighborhood} (see, e.g.,~\cite{B12,NO13}).
Note that if $k=0$, then \textsc{Rob Rec SP}  is equivalent to   \textsc{SP} problem under the arc costs
 $C_e +\hat{c}_e+\Delta_e$, $e\in A$.   Thus,
from now on, we assume that~$k\geq 1$.

\section{Recoverable  (robust) shortest path problem}
We start with a simple observation (see~\cite{B12}), which 
 simplifies  considerably the \textsc{Rob Rec SP}  
under scenario set~$\Uset$, namely,
$\max_{S\in \Uset}\min_{Y\in \Phi(X,k)} \sum_{e\in Y} c_e^S=
\min_{Y\in \Phi(X,k)} \sum_{e\in Y} (\hat{c}_e+\Delta_e)$.
Hence  \textsc{Rob Rec SP} is equivalent 
 the following problem
 \begin{equation}
\textsc{Rec SP}: \; \min_{X\in \Phi, Y\in \Phi(X,k)} \left(\sum_{e\in X} C_e +  \sum_{e\in Y} \overline{c}_e\right)
\label{rsp}
\end{equation}
 called
\emph{recoverable   shortest path problem}, where  $\overline{c}_e$ stands for $\hat{c}_e+\Delta_e$.
Thus, from now on, we will study the \textsc{Rec SP} problem instead of 
 the \textsc{Rob Rec SP} under~$\Uset$.

In~\cite{B12}, it has been proved that the  \textsc{Rec SP} (equivalently \textsc{Rob Rec SP}) problem 
with~$\Phi(X,k)$
 is strongly 
NP-hard and not approximable in  general digraphs, even
if the recovery parameter~$k=2$.

In this section, we show that \textsc{Rec SP} with can be solved in polynomial time if the input graph $G$ is an acyclic multigraph. Observe that we can drop the assumption about  nonnegativity of the arc costs for this class of graphs. We first assume that~$G$ is layered, and then we will generalize the result for arbitrary acyclic multigraphs. We finish with
arc series-parallel multidigraphs. 

\subsection{Layered multidigraph}
In a layered multigraph $G=(V,A)$ the set of nodes $V$ is partitioned into $H$ disjoint subsets, i.e. $V=
  V_1\cup \cdots \cup V_H$, and the arcs go only from nodes in $V_h$ to nodes in $V_{h+1}$ for each
  $h\in [H-1]$. We can assume  that $s\in V_1$ and $t\in V_H$.
\begin{theorem}
\textsc{Rec SP} with~$\Phi(X,k)$ in  
 a layered multidigraph~$G=(V,A)$ can solved in $O(|A||V|+|V|^2k)$ time.
 \label{tred}
\end{theorem}
\begin{proof}
We show a polynomial transformation from \textsc{Rec SP} in a layered multidigraph to 
  the Constrained Shortest Path Problem (\textsc{CSP} for short) in an acyclic multidigraph. In the \textsc{CSP} problem, 
  we are given an acyclic multidigraph~$G=(V,A)$  with $s\in V $ and $t\in V$.
  Each arc $e\in A$ has a cost~$c_e$ and a transition time~$t_e$.  A positive total transition time limit~$T$ is specified. 
  We seek an $s$-$t$ path $\pi$ in $G$ that minimizes the
total cost, subject to not exceeding the transition time limit~$T$.
The  \textsc{CSP} can be solved in $O(|A|T)$ time~\cite{H92}.

 We now build an acyclic multidigraph $G'=(V',A')$ in the corresponding instance of \textsc{CSP} as follows.
 We fix $V'=V$ and  $A'$ contains two types of arcs labelled as $e^{(0)}$ and $e^{(1)}$. Namely,
 for each pair of nodes $i\in V_{h}$ and $j\in V_{h+1}$, $h\in [H-1]$, such that there exists at least one arc 
 from~$i$ to~$j$ (note that $G$ can be multidigraph), 
 we add to $A'$ arc $e^{(0)}=(i,j)$ with the cost $c_{e^{(0)}}=\min_{e=(u,v)\in A\,:\,u=i, v=j} \{C_e+\overline{c}_e\}$
 and the transition time~$t_{e^{(0)}}=0$. 
 For each pair of nodes $i\in V_g$ and $j\in V_h$, $g,h\in [H]$, such that the node~$j$ is reachable from~$i$ and
 $1\leq h-g\leq k$,
 we add to $A'$ arc $e^{(1)}=(i,j)$ with the cost $c_{e^{(1)}}$ being the sum
 of the cost of a shortest path~$X^*_{ij}$ from~$i$ to~$j$ under~$C_e$, $e\in A$, and 
 the cost of a shortest path~$Y^*_{ij}$ from~$i$ to~$j$ under~$\overline{c}_e$, $e\in A$.
  The the transition time~$t_{e^{(1)}}=|Y^*_{ij}\setminus X^*_{ij}|$. Since $G$ is layered
 $|X^*_{ij}|=|Y^*_{ij}|= h-g$ and thus
 $|Y^*_{ij}\setminus X^*_{ij}|\leq  h-g$.
 Finally, we set $T=k$. 
 It is not difficult to show that 
 there is a pair of  paths $X\in \Phi$ and $Y\in \Phi(X,k)$ in~$G$,
 feasible to~\textsc{Rec SP}, with the cost  $\sum_{e\in X} C_e+\sum_{e\in Y}\overline{c}_e\leq UB$ 
 if and only if  there is a path~$\pi\in \Phi$ in~$G'$,  feasible to~\textsc{CSP},
  such that $\sum_{e\in \pi} c_e\leq UB$, where $UB\in \Rset$.
 The graph $G'$ can be constructed in $O(|A||V|)$ time.
 The number of arcs in $G'$ is $O(|V|^2)$, so the corresponding \textsc{CSP} problem can be solved in $O(|V|^2 k)$ time. Hence the overall running time is $O(|A||V|+|V|^2k)$.
  \end{proof}

\subsection{General acyclic multidigraph}

\begin{theorem}
\textsc{Rec SP} with~$\Phi(X,k)$ in  
 an acyclic  multidigraph~$G=(V,A)$ can solved in $O(|V|^2 |A| k^2)$ time.
\end{theorem}
\begin{proof}
Similarly as in the proof of Theorem~\ref{tred}, we give
a polynomial  time reduction from  \textsc{Rec SP} to
\textsc{CSP} in an acyclic multidigraph.
We need to extend the construction of the graph~$G'$, since
in general acyclic multidigraphs, disjoint paths $X^*_{ij}$  and $Y^*_{ij}$ from node~$i$ to node~$j$ can have various cardinalities, 
so cheaper recoverable action can require choosing more arcs.
 We fix $V'=V$. The set of arcs~$A'$ can now contain~$k+1$ types of arcs labeled as  $e^{(0)},\dots,e^{(k)}$.
 For each pair of nodes $i,j\in V$, $i\not= j$, such that there exists at least one arc 
 from~$i$ to~$j$, 
 we add to $A'$ the arc $e^{(0)}=(i,j)$ with the cost $c_{e^{(0)}}=\min_{e=(k,l)\in A\,:\,k=i, l=j} \{C_e+\overline{c}_e\}$
 and the transition time~$t_{e^{(0)}}=0$. 
 Consider a pair of nodes $i,j\in V$, $i\not=j$, such that $j$ is reachable from $i$. We first compute a path from $i$ to $j$ having the minimum number of arcs $L^*_{ij}$. If $L_{ij}^*>k$, then we do nothing. Otherwise, we find a shortest path $X^*_{ij}$ from $i$ to $j$ for the costs $C_e$. Then, for each $l=L_{ij}^*,\dots,k$, we find a shortest path $Y^{*l}_{ij}$ for the costs $\overline{c}_e$, $e\in A$, subject to the condition that $Y^{*l}_{ij}$ has at most $l$ arcs. Observe that $Y^{*l}_{ij}$ can be found in $O(|A|l)$ time by solving 
the  \textsc{CSP} problem with arc transition times equal to~1 and $T=l$.
 We add to $A'$ arc $e^{(l)}=(i,j)$ with the cost $c_{e^{(l)}}=\sum_{e\in X^*_{ij}} C_e+
 \sum_{e\in Y^{l*}_{ij}} \overline{c}_e$
 and with the transition time~$t_{e^{(l)}}=l$.
 Finally, we set $T=k$. 
 It is  not hard to show that
 there is a pair of  paths $X\in \Phi$ and $Y\in \Phi(X,k)$ in~$G$,
 feasible to~\textsc{Rec SP}, with the cost  $\sum_{e\in X} C_e+\sum_{e\in Y}\overline{c}_e\leq UB$ 
 if and only if  there is a path~$\pi\in \Phi$ in~$G'$,  feasible to~\textsc{CSP},
  such that $\sum_{e\in \pi} c_e\leq UB$, where $UB\in \Rset$.
 The graph $G'$ can be built in $O(|V|^2|A|k^2)$ time. The number of arcs in $G'$ is $O(|V|^2k)$ and the resulting 
 \textsc{CSP} problem can be solved in $O(|V|^2k^2)$ time. Thus the overall running time is $O(|V|^2 |A| k^2)$.  
\end{proof}

\subsection{Arc series-parallel multidigraph (ASP)}

In~\cite{B12}, {B{\"u}sing gave an idea of a polynomial algorithm for 
\textsc{Rec SP} with~$\Phi(X,k)$ in  an
 ASP multidigraph together with the theorem resulting from it~\cite[Theorem~5]{B12}.
 
 We now propose a complete $O(|A|k^2)$-algorithm for  ASP multidigraphs, which turned out to be far from trivial,
 with a proof of correctness. The algorithm is as follows.
We first recall that an ASP multidigraph~$G$  is associated with a rooted binary
tree~$\mathcal{T}$ called \emph{the binary decomposition tree
of~$G$} (see, e.g.,~\cite{VTL82}).
 Each leaf of the tree represents an arc in~$G$. Each
internal node $\sigma$ of~$\mathcal{T}$ is labeled  $\mathtt{S}$
or $\mathtt{P}$ and corresponds to the series or parallel composition
in $G$. Every node $\sigma$ of $\mathcal{T}$ corresponds to  an
ASP subgraph of~$G$, denoted by $G_{\sigma}$, defined by the subtree rooted at~$\sigma$.
The root of~$\mathcal{T}$ represents the input ASP multidigraph~$G$.
For each ASP subgraph~$G_{\sigma}$, we store the following three crucial pieces of information. Namely,
the first one, the cost of a shortest path from  the source  to  the sink in~$G_{\sigma}$ with under costs~$C_e$, $e\in A$:
\begin{equation}
C_{G_{\sigma}}=\min_{X\in \Phi_{G_{\sigma}}} \sum_{e\in X} C_e.
\label{locost}
\end{equation}
Here and subsequently, $\Phi_{G_{\sigma}}$ denotes the set of all paths from the source to the sink in~$G_{\sigma}$.
The second one, the $k$-element array~$\overline{\pmb{c}}_{G_{\sigma}}$, whose $l$th element~$\overline{\pmb{c}}_{G_{\sigma}}[l]$, $l=1,\ldots,k$,
is the cost of a shortest path from  the source  to  the sink in~$G_{\sigma}$  under the costs~$\overline{c}_e$, $e\in A$,
that uses exactly~$l$ arcs -- if such a path exists:
\begin{equation}
\overline{\pmb{c}}_{G_{\sigma}}[l]=
\begin{cases}
\displaystyle \min_{Y\in \Phi_{G_{\sigma}} \,:\,|Y|=l} \sum_{e\in Y}\overline{c}_e &\text{if $\{Y\in \Phi_ {G_{\sigma}}\,:\,|Y|=l\}\not=\emptyset$,}\\
+\infty&\text{otherwise}.
\end{cases}
\label{upcost}
\end{equation}
The third, the $(k+1)$-element array~$\pmb{c}^*_{G_{\sigma}}$, whose  $l$th element~$\pmb{c}^*_{G_{\sigma}}[l]$,
$l=0,\ldots,k$,
stores the cost of an optimal pair $X^*,Y^*\in \Phi_{G_{\sigma}}$ such that $|Y^*\setminus X^*|=l$ -- 
if such  paths exist:
\begin{equation}
\pmb{c}^*_{G_{\sigma}}[l]=
\begin{cases}
\displaystyle\min_{X,Y\in \Phi_{G_{\sigma}}\,:\,|Y\setminus X|=l} \left(\sum_{e\in X} C_e +  \sum_{e\in Y} \overline{c}_e\right)&
\text{if $\{X\in\Phi_{G_{\sigma}} \,:$}\\
&\text{$(\exists Y\in \Phi_{G_{\sigma}}) (|Y\setminus X|= l)\}\not=\emptyset$}\\
+\infty&\text{otherwise}.
\end{cases}
\label{optcost}
\end{equation}
For each leaf node~$\sigma$ of the tree~$\mathcal{T}$, in this case, the  subgraph~$G_{\sigma}$ consists of
the single arc~$e$, $e\in A$, the initial values in~(\ref{locost}),  (\ref{upcost}) and~(\ref{optcost}) are as follows:
$C_{G_{\sigma}}= C_e$;
$\overline{\pmb{c}}_{G_{\sigma}}[1]=\overline{c}_e,\;\overline{\pmb{c}}_{G_{\sigma}}[l]=+\infty, l=2,\ldots,k$;
$\pmb{c}^*_{G_{\sigma}}[0]=C_e+\overline{c}_e,\;\pmb{c}^*_{G_{\sigma}}[l]=+\infty, l=1,\ldots,k$.
By traversing the tree~$\mathcal{T}$ from the leaves to the root, we recursively 
construct each  ASP subgraph~$G_{\sigma}$ corresponding to the internal node~$\sigma$ and
compute 
the values in~(\ref{locost}),  (\ref{upcost}) and~(\ref{optcost}) associated to~$G_{\sigma}$,
 depending on the label of the node~$\sigma$.
If $\sigma$ is marked~$\mathtt{P}$, then 
we call Algorithm~\ref{parallel},
otherwise we  call Algorithm~\ref{series}, on the subgraph~$G_{\mathrm{left}(\sigma)}$ and $G_{\mathrm{right}(\sigma)}$,
where $\mathrm{left}(\sigma)$ and $\mathrm{right}(\sigma)$ are children of~$\sigma$ in~$\mathcal{T}$.
If $\sigma$ is the root of~$\mathcal{T}$, then  $G_{\sigma}=G$ and the array~$\pmb{c}^*_{G_{\sigma}}$
contains information about the cost 
of an optimal 
pair of paths~$X^*,Y^*\in \Phi$, 
of the optimal solution to \textsc{Rec SP} in~$G$, which is equal to $\min_{0\leq l\leq k} \pmb{c}^*_{G_{\sigma}}[l]$.
For simplicity of the presentation, we have shown only how to compute the costs stored in~$C_{G_{\sigma}}$, 
$\overline{\pmb{c}}_{G_{\sigma}}$ and~$\pmb{c}^*_{G_{\sigma}}$. The associated paths can easily
be reconstructed.
\begin{algorithm}
\begin{small}
Let $G_{\sigma}$ be the
 parallel composition of $G_{\mathrm{left}(\sigma)}$ and $G_{\mathrm{right}(\sigma)}$\;\label{parallel1}
$\pmb{c}^*_{G_{\sigma}}[0]\leftarrow \min\{\pmb{c}^*_{G_{\mathrm{left}(\sigma)}}[0], \pmb{c}^*_{G_{\mathrm{right}(\sigma)}}[0]\}$;
$C_{G_{\sigma}}\leftarrow \min\{C_{G_{\mathrm{left}(\sigma)}},C_{G_{\mathrm{right}(\sigma)}}\}$\;
 \label{parallel2}
\For{$l=1$ \emph{\KwTo} $k$}{
$\pmb{c}^*_{G_{\sigma}}[l]\leftarrow \min\{\pmb{c}^*_{G_{\mathrm{left}(\sigma)}}[l], \pmb{c}^*_{G_{\mathrm{right}(\sigma)}}[l], 
C_{G_{\mathrm{left}(\sigma)}}+
\overline{\pmb{c}}_{G_{\mathrm{right}(\sigma)}}[l],  C_{G_{\mathrm{right}(\sigma)}}+\overline{\pmb{c}}_{G_{\mathrm{left}(\sigma)}}[l]\}$
\label{parallel4}\;
$\overline{\pmb{c}}_{G_{\sigma}}[l]\leftarrow \min\{\overline{\pmb{c}}_{G_{\mathrm{left}(\sigma)}}[l],\overline{\pmb{c}}_{G_{\mathrm{right}(\sigma)}}[l]\}$ \label{parallel5}
}
\Return{$(G_{\sigma}, c^*_{G_{\sigma}}, C_{G_{\sigma}}, \overline{c}_{G_{\sigma}})$}
  \caption{$\mathtt{P}(G_{\mathrm{left}(\sigma)},\pmb{c}^*_{G_{\mathrm{left}(\sigma)}},\overline{\pmb{c}}_{G_{\mathrm{left}(\sigma)}}, G_{\mathrm{right}(\sigma)}, \pmb{c}^*_{G_{\mathrm{right}(\sigma)}},C_{G_{\mathrm{right}(\sigma)}},\overline{\pmb{c}}_{G_{\mathrm{right}(\sigma)}})$}
 \label{parallel}
\end{small} 
\end{algorithm}
\begin{algorithm}
\begin{small}
Let $G_{\sigma}$ be the
series composition of  $G_{\mathrm{left}(\sigma)}$ and $G_{\mathrm{right}(\sigma)}$\;
\lFor{$l=0$ \emph{\KwTo} $k$}{
$\pmb{c}^*_{G_{\sigma}}[l]\leftarrow \min_{0\leq j\leq l}\{\pmb{c}^*_{G_{\mathrm{left}(\sigma)}}[j]+\pmb{c}^*_{G_{\mathrm{right}(\sigma)}}[l-j]\}$
}\label{series3}
$C_{G_{\sigma}}\leftarrow C_{G_{\mathrm{left}(\sigma)}}+C_{G_{\mathrm{right}(\sigma)}}$;
$\overline{\pmb{c}}_{G_{\sigma}}[1]\leftarrow +\infty$\label{series4}\;
\lFor{$l=2$ \emph{\KwTo} $k$}{
$\overline{\pmb{c}}_{G_{\sigma}}[l]\leftarrow \min_{1\leq j< l}\{\overline{\pmb{c}}_{G_{\mathrm{left}(\sigma)}}[j]+\overline{\pmb{c}}_{G_{\mathrm{right}(\sigma)}}[l-j]\}$
} \label{series7}
\Return{$(G_{\sigma}, c^*_{G_{\sigma}}, C_{G_{\sigma}}, \overline{c}_{G_{\sigma}})$}
  \caption{$\mathtt{S}(G_{\mathrm{left}(\sigma)},\pmb{c}^*_{G_{\mathrm{left}(\sigma)}},\overline{\pmb{c}}_{G_{\mathrm{left}(\sigma)}}, G_{\mathrm{right}(\sigma)}, \pmb{c}^*_{G_{\mathrm{right}(\sigma)}},C_{G_{\mathrm{right}(\sigma)}},\overline{\pmb{c}}_{G_{\mathrm{right}(\sigma)}})$}
 \label{series}
\end{small} 
\end{algorithm}

We now prove the correctness of the above algorithm.
\begin{theorem}
\textsc{Rec SP} with~$\Phi(X,k)$ in  an
 ASP multidigraph~$G=(V,A)$ can solved in $O(|A|k^2)$ time.
\end{theorem}
\begin{proof}
We first prove that for each node~$\sigma$  of the decomposition tree~$\mathcal{T}$,
the graph~$G_{\sigma}$ is the ASP subgraph of~$G$ and the 
corresponding  costs~(\ref{locost}),  (\ref{upcost}) and~(\ref{optcost})
are correctly computed by the algorithm. The first part of the claim is due to~\cite{VTL82}.
If $|A|=1$, then the claim trivially holds.
Assume that $|A|\geq 2$. The rest 
 proof of the claim is by induction on the number of iterations in which   the algorithm
traverses the tree~$\mathcal{T}$ from the leaves to the root.

The \emph{base case}  ($|A|=2$) is trivial to verify, since
the node~$\sigma$ is the root of~$\mathcal{T}$ and has two  leaves
corresponding to 
two single arc subgraphs~$G_{\mathrm{left}(\sigma)}$ and~$G_{\mathrm{right}(\sigma)}$ and 
the algorithm proceeds for only one iteration.
 For the  \emph{induction step}, let~$\sigma$ be an internal  node of~$\mathcal{T}$ with two leaves  $\mathrm{left}(\sigma)$ and $\mathrm{right}(\sigma)$.
 By the induction hypothesis, $G_{\mathrm{left}(\sigma)}$ and~$G_{\mathrm{right}(\sigma)}$ are ASP subgraphs and
the values~(\ref{locost}),  (\ref{upcost}) and~(\ref{optcost}) associated  with them
 are correctly computed. We now need to consider two cases that depend on the label of~$\sigma$.
 The first case (the label~$\mathtt{P}$), when Algorithm~\ref{parallel} is called for 
$G_{\mathrm{left}(\sigma)}$ and~$G_{\mathrm{right}(\sigma)}$.
Atfer the parallel composition of $G_{\mathrm{left}(\sigma)}$ and~$G_{\mathrm{right}(\sigma)}$,
the resulting subgraph~$G_{\sigma}$ is an ASP subgraph. 
Note that, $\Phi_{G_{\sigma}}= \Phi_{\mathrm{left}(\sigma)} \cup \Phi_{G_{\mathrm{left}(\sigma)}}$,
where $\Phi_{G_{\mathrm{left}(\sigma)}}\cap \Phi_{G_{\mathrm{left}(\sigma)}}=\emptyset$.
Hence, we immediately get the costs in lines~\ref{parallel2},  \ref{parallel4} and~\ref{parallel5}
(see Algorithm~\ref{parallel}) are equal to the ones~(\ref{locost}), (\ref{upcost})  and~(\ref{optcost}).
The costs $\pmb{c}^*_{G_{\sigma}}[l]$, for $l=1,\ldots,k$, (see line~\ref{parallel5}),
are correctly computed as well.
Indeed, the first two terms are obvious. The last two take into account the cases, when the
optimal paths $X^*\in \Phi_{G_{\sigma}}$ and $Y^*\in \Phi_{G_{\sigma}}$ in~(\ref{optcost}) are such that
$X^*\in \Phi_{G_{\mathrm{left}(\sigma)}}$ and $Y^*\in \Phi_{G_{\mathrm{right}(\sigma)}}$ or
$X^*\in \Phi_{G_{\mathrm{right}(\sigma)}}$ and $Y^*\in \Phi_{G_{\mathrm{left}(\sigma)}}$.
Consider the second case (the label~$\mathtt{S}$) when Algorithm~\ref{series} is called.
The subgraph~$G_{\sigma}$, after  the series composition of $G_{\mathrm{left}(\sigma)}$ and~$G_{\mathrm{right}(\sigma)}$,
is an ASP subgraph and there is the node~$v$ in $G_{\sigma}$ such that $v$
the node after identifying  the sink of~$G_{\mathrm{left}(\sigma)}$
and the source of~$G_{\mathrm{right}(\sigma)}$.
Thus, every path $\pi\in \Phi_{G_{\sigma}}$ must traverse the node~$v$ and $\pi$ is 
the concatenation of $\pi_1\in  \Phi_{G_{\mathrm{left}(\sigma)}}$ and 
$\pi_2\in  \Phi_{G_{\mathrm{right}(\sigma)}}$, $\pi=(\pi_1,v,\pi_2)$.
From the above it follows that
for each $l=0,\ldots,k$, the cost
$\pmb{c}^*_{G_{\sigma}}[l]$ in $G_{\sigma}$
is the sum of its optimal counterparts, respectively, in  $G_{\mathrm{left}(\sigma)}$ and~$G_{\mathrm{right}(\sigma)}$,
for some $0\leq j^*\leq l$,  $l=j^*+l-j^*$. It is enough to find such~$j^*$ for each~$l$.
A similar argument holds for the costs
$\overline{\pmb{c}}_{G_{\sigma}}[l]$ for $l=2,\ldots,k$.
Therefore, the costs in lines~\ref{series3}, \ref{series4} and~\ref{series7}  (see Algorithm~\ref{series}) are
equal to the ones~(\ref{locost}), (\ref{upcost})  and~(\ref{optcost}).
This proves the claim.

Let $\sigma$  be  the root of~$\mathcal{T}$, then  $G_{\sigma}=G$. It immediately  follows  from the claim,
that  the array~$\pmb{c}^*_{G_{\sigma}}$ contains the optimal costs~(\ref{optcost}) for $l=0,\ldots,k$.
Thus, 
the cost 
of an optimal solution~$X^*,Y^*\in \Phi$ to \textsc{Rec SP} in~$G$ is equal to $\min_{0\leq l\leq k} \pmb{c}^*_{G_{\sigma}}[l]$.

Let us analyze the running time of the Algorithm. The binary
decomposition tree $\mathcal{T}$ of~$G$ can be constructed in $O(|A|)$ time~\cite{VTL82}.
The initialization can be done in $O(|A|)$ time.
The root of $\mathcal{T}$ can be
reached in $O(|A|)$ time.  Algorithms~\ref{parallel} and~\ref{series}
require~$O(k^2)$ time.
Hence, the running time of the  algorithm is $O(|A|k^2)$.
\end{proof}

\subsubsection*{Acknowledgements}
This work was supported by
 the National Science Centre, Poland, grant 2022/45/B/HS4/00355.

%
%
%

\end{document}